\newif\ifabstract
\abstracttrue
 \abstractfalse 
\newif\iffull
\ifabstract \fullfalse \else \fulltrue \fi

\documentclass[11pt]{article}

\usepackage{amsfonts}
\usepackage{amssymb}
\usepackage{amstext}
\usepackage{amsmath}
\usepackage{xspace}
\usepackage{amsthm}
\usepackage{graphicx}
\usepackage{url}
\usepackage{graphics}
\usepackage{colordvi}
\usepackage{colordvi}
\usepackage{subfigure}
\usepackage[numbers]{natbib}

\usepackage{enumitem}
\usepackage{pifont}
\usepackage{tikz}
\usepackage{xcolor}

\definecolor{goldButNotOld}{HTML}{F26035}
\definecolor{orangeLikeAFamousPerson}{HTML}{D35400}

\usepackage{booktabs}
\usepackage{adjustbox}
\usepackage{xspace}
\usepackage{subcaption}
\usepackage{enumitem}
\usepackage{pifont}
\usepackage{tikz}
\usepackage{xcolor}

\usepackage{array}
\usepackage{graphicx}
\usepackage{caption}
\usepackage{arydshln}
\usepackage{bbold}
\usepackage{mathtools}
\usepackage{multirow}
\usepackage[ruled,linesnumbered,vlined]{algorithm2e}
\newtheorem{theorem}{Theorem}
\newtheorem{remark}{\textbf{Remark}}[section]
\newtheorem{assumption}{\textbf{Assumption}}

\allowdisplaybreaks 

\newcommand{\TORA}{\texttt{TORA}\xspace}
\newcommand{\cd}{\texttt{CD}\xspace}

\newcommand{\probName}{\texttt{DCP}\xspace}
\newcommand{\algName}{\texttt{LARA}\xspace}
\newcommand{\E}{\mathbb{E}}

\advance \oddsidemargin by -0.8in
\newcommand{\myparskip}{3pt}
\parskip \myparskip

\usepackage[most]{tcolorbox}

\newtcolorbox{visionbox}[2][]{%
    colback=cyan!10,  
    coltitle=black,
    colframe=cyan!60,  
    fonttitle=\bfseries,
    title=#2, 
    rounded corners=southeast,
    boxrule=0pt,
    enhanced,
    drop fuzzy shadow,
    #1, 
    top=3pt, bottom=2pt, left=3pt, right=3pt
}

\begin{document}

\title{Near-Optimal Emission-Aware Online Ride Assignment Algorithm for Peak Demand Hours}

\author{
Ali~Zeynali\thanks{University of Massachusetts Amherst. Email: {\tt azeynali@cs.umass.edu}.} \and
Mahsa~Sahebdel\thanks{University of Massachusetts Amherst. Email: {\tt msahebdelala@umass.edu }.}\and 
Noman~Bashir\thanks{MIT. Email: {\tt nbashir@mit.edu }.}\and 
Ramesh~K.~Sitaraman\thanks{University of Massachusetts Amherst \& Akamai Technologies. Email: {\tt ramesh@cs.umass.edu}.}\and
Mohammad~Hajiesmaili\thanks{University of Massachusetts Amherst. Email: {\tt hajiesmaili@cs.umass.edu}.} \and
}

\date{} 

\begin{titlepage}
\maketitle

\thispagestyle{empty}

\begin{abstract}
Ridesharing has experienced significant global growth over the past decade and is becoming an integral component of modern transportation systems. However, despite their benefits, ridesharing platforms face fundamental inefficiencies that contribute to negative environmental impacts. A prominent source of such inefficiency is the ``deadhead miles''—the distance traveled by vehicles without passengers between trips—which accounts for a substantial portion of carbon emissions. This issue becomes especially severe during high-demand periods, when the volume of ride requests exceeds the available driver supply, leading to suboptimal rider-to-driver assignments, longer deadhead trips, and increased emissions. Although limiting these unproductive miles can reduce emissions, doing so may increase passenger wait times due to limited  driver availability, thereby degrading the overall service experience. As ridesharing continues to scale, there is a critical need for environment-aware solutions that jointly minimize emissions and maintain high-quality service, particularly in terms of rider wait times.

In this paper, we introduce \algName, an online rider-to-driver assignment algorithm that dynamically adjusts the maximum allowable distance between rider and drivers and assigns ride requests accordingly. While \algName is applicable in general settings, it is particularly effective during peak demand periods, achieving reductions in both emissions and wait times. We provide theoretical guarantees showing that \algName achieves near-optimal performance in online environments, with respect to an optimal offline benchmark. Beside our theoretical analysis, our empirical evaluations on both synthetic and real-world datasets show that \algName achieves up to a $34\%$ reduction in carbon emissions and up to a $50\%$ decrease in rider wait times, compared to state-of-the-art baselines. While prior work has explored emission-aware ride assignment, \algName is, to our knowledge, the first algorithm to offer both rigorous theoretical guarantees and strong empirical performance.
\end{abstract}

\end{titlepage}

\section{Introduction}
\label{sec:intro}

Ridesharing services have revolutionized urban mobility by providing convenient, on-demand transportation and have seen widespread global adoption~\cite{wenzel2019travel}. The global ridesharing market is projected to reach \$212B by 2029, with an estimated 2.3B users~\cite{statista_ride_hailing}, representing a 167\% increase in market size and a 53\% increase in user base compared to 2020. Although these services were initially perceived as environmentally friendly—promising reductions in emissions and congestion—recent studies have shown otherwise. A key contributor to the increased emissions is \emph{deadheading}: miles driven by drivers without passengers. These deadhead miles can significantly increase overall energy consumption, and in fact, a rideshared trip has been shown to generate approximately 47\% more CO2 emissions than a comparable private car trip~\cite{kontou2020reducing,sahebdel2024holistic,reuters,sahebdel2025lead}. Given their scale and growing influence on modern transportation systems, understanding and mitigating the environmental impacts of ridesharing has become increasingly critical.

During peak demand hours, ridesharing platforms face challenges like driver selectiveness and trip cancellations, leading to assignment inefficiencies~\cite{nanda2020balancing}. Even centralized assignment algorithms struggle to optimize both wait times and emissions when driver availability is low. Simple assignment strategies, such as first-come, first-served, can increase deadhead miles and worsen user experience. Addressing these inefficiencies requires intelligent online algorithms that balance service quality and sustainability. Moreover, accurately predicting the exact timing and magnitude of peak demand introduces significant uncertainty, which can further reduce the effectiveness of even advanced assignment strategies~\cite{goodwin2016pattern,khojaste2025electricity,khojaste2024quantile}.

While recent studies have explored methods to reduce deadhead miles and their associated emissions~\cite{kontou2020reducing}, and others have aimed to improve service quality by minimizing wait times~\cite{sahebdel2023data, sahebdel2024holistic,sahebdel2025lead}, many of these approaches lack performance guarantees across varying conditions. This issue becomes especially critical during high-demand periods when driver availability is low. In such cases, online algorithms without performance guarantees can fall significantly short compared to the optimal offline algorithms, which benefit from knowing future inputs. This underscores the need for new strategies that ensure environmentally sustainable ridesharing systems with reliable emissions reductions, even during peak demand, while also maintaining high service quality with respect to wait times.

\begin{figure}[t!]
    \centering
    \includegraphics[width=0.65\linewidth]{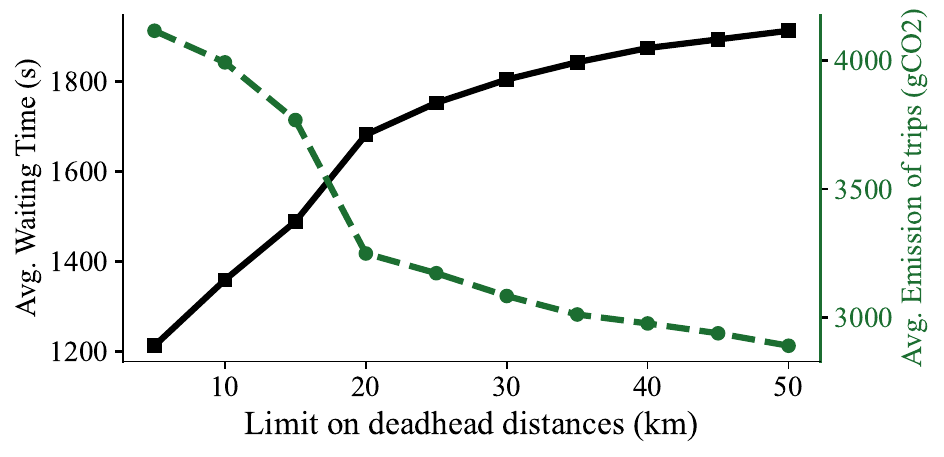}
    \caption{Adjusting the maximum allowable deadhead distance could balance the trade-off between average emission and rider waiting times within the ridesharing system.}
    \label{fig:motivation_figure}
\end{figure}

Developing such strategies with theoretical guarantees in ridesharing optimization is inherently complex, as it requires managing dynamic decision-making processes involving both current and future drivers and riders. More specifically, a major challenge lies in the uncertainty of future ride requests' timing and locations. Assigning a driver to a rider not only affects the current ride but also determines when and where that driver will be available again, influencing the deadhead distance and emissions of subsequent rides. This causes long-term changes in the online algorithm's actions and overall performance. Furthermore, designing an online algorithm that balances the platform's sustainability goals—such as emission reduction—with rider preferences, like minimizing wait times, adds another layer of complexity. Real-world factors, such as fluctuating urban traffic conditions, further complicate efforts to ensure consistent, optimal performance in practice.

\noindent\textbf{Contributions.}
Due to the above complexities, formulating a tractable multi-objective ride assignment problem that lends itself to rigorous algorithmic design and theoretical analysis poses significant challenges. As an alternative, inspired by the empirical findings presented in Section~\ref{sec:DCP_formulation} (see Figure~\ref{fig:motivation_figure}), we propose an alternative approach to control the trade-off between emissions and waiting times, framing it as a Deadhead Control Problem (\probName). We model the objective of \probName as a weighted sum of two factors: the expected carbon reduction and the rate of ride assignment which directly impacts the rider waiting time. As a ride assignment problem, this approach implicitly balances the trade-off between emissions reduction and rider wait times and, more importantly, is tractable for rigorous algorithm design and performance analysis.

In Section~\ref{sec:LARA}, we solve \probName by proposing
 a \underline{L}yapunov-based \underline{A}lgorithm for \underline{R}ide \underline{A}ssignment (\algName), an online algorithm designed to reduce carbon emissions in ridesharing platforms while minimizing rider wait times.  \algName's  benefits over previous state-of-the-art algorithms are especially pronounced during high-demand hours, when a naive assignment approach could lead to long queues of unassigned requests leading to long waiting times.  \algName dynamically adjusts the upper bound on deadhead distances for assigned drivers based on real-time conditions. The algorithm's decisions are influenced by the number of ride requests in the assignment queue (i.e., unassigned requests waiting to be dispatched), and it uses a tunable hyperparameter that allows for performance optimization under such dynamic conditions. 

In Section~\ref{sec:thm_analysis}, we analyze the theoretical performance of \algName. 
We show that the objective value obtained by \algName for \probName is within a bounded distance from the optimal solution when extra (Theorem~\ref{thm:perf_analysis}) and no additional (Theorem~\ref{thm:perf_analysis_noisy}) information about environment is provided in advance. We also show that this bound approaches zero when \algName is not constrained by the assignment queue length, which essentially shows the optimality of the algorithms under high-demand and rush hour periods.
    
Finally, in Section~\ref{sec:experiments}, we empirically evaluate the performance of \algName using both synthetic and real-world datasets. Comparing its performance against the existing emission-aware ride assignment algorithm \TORA~\cite{sahebdel2024holistic} as well as additional heuristic methods. Our results indicate significant reductions in both emissions and waiting times during high-demand periods. For instance, compared to state-of-the-art emission and waiting time aware baselines, \algName achieves up to $34\%$ reduction in average emissions on the synthetic dataset (Figure~\ref{fig:synthetic_batchDur}) and up to $13.9\%$ on the real-world dataset (Figure~\ref{fig:dataset_batchDur}), consistently outperforming competing algorithms across various scenarios. Furthermore, our analysis of \algName reveals a trade-off between reducing emissions and ensuring fair ride assignments among different drivers. This trade-off in the algorithm design could be leveraged by the ridesharing platform to balance between sustainability and customer satisfaction goals in runtime.

\section{Problem Formulation}
\label{sec:DCP_formulation}


The ridesharing platform consists of $M$ drivers and $N$ ride requests, where $N$ is unknown in advance and ride requests arrive sequentially over time. Each ride request, indexed by $n$, includes a request time, a pickup location, and a drop-off location. Each driver, indexed by $m$, operates a vehicle that emits $e_m$ grams of CO2 per unit distance and can serve only one ride at a time. Ride requests arrive sequentially, and upon each new request, the online assignment algorithm updates the drivers' locations and statuses before assigning a ride. Once assigned, a driver picks up the rider and remains unavailable until the drop-off is complete.  

The objective is to minimize two key factors: (1) total carbon emissions and (2) average rider waiting time. Carbon emissions for a trip depend on both the deadhead distance and the trip distance, as well as the vehicle’s emission efficiency. Rider waiting time consists of the time period the rider waits until the assignment happens and the time until the driver arrives at the pickup location. While shorter waiting times benefit riders, minimizing both waiting time and emissions requires trade-offs, making it challenging to optimize both simultaneously.

Previous studies \cite{sahebdel2024holistic,sahebdel2023data, sahebdel2025lead} highlight the trade-off between minimizing deadhead emissions and reducing passenger wait times in ridesharing platforms. Assigning passengers to drivers with low-emission vehicles can reduce emissions but may increase wait times due to longer pickup distances. Conversely, imposing a strict limit on the maximum allowable distance between passengers and assigned drivers can shorten wait times but reduces the chances of finding nearby low-emission vehicles, potentially increasing overall  emissions.  To better capture the trade-off between emission and waiting times in the ridesharing platform, we performed a simple analysis over the ride requests posted during first week of December 2016 from RideAustin dataset~\cite{rideaustin-dataset} (see Section~\ref{sec:experiments} for details). We imposed a cap on the deadhead distance of eligible drivers and assigned each rider to the driver with the lowest emissions. Then, we analyzed the impact of this limit on deadhead distances on the performance of ridesharing system. The results in Figure~\ref{fig:motivation_figure} indicate that instead of directly modifying ride assignment strategies based on platform status, the online algorithm can adjust the maximum allowable deadhead distance for assigned rides, balancing the trade-off between average emissions and rider waiting times. This insight motivates an alternative approach to modeling ridesharing optimization with a more tractable theoretical analysis. In the following, we introduce a novel problem formulation that captures the objectives and optimization challenges of ridesharing systems.

\subsection{\probName: Deadhead Control Problem}
\label{sec:DCP}
Motivated by the insights from our preliminary empirical results in Figure~\ref{fig:motivation_figure}, and to address the challenges of online algorithm design for emission-aware ride assignment, we introduce the \emph{Online Deadhead Control Problem}, referred to as \probName. This problem offers a new perspective on managing emissions and deadhead distances in ridesharing platforms. In \probName, ride assignments are made in batches, and the algorithm determines the maximum allowable deadhead distance for each upcoming batch. The objective is to jointly minimize total emissions and passenger wait times over the entire time horizon. Each ride request is then matched to a driver within this limit, with a focus on reducing emissions. The resulting online ride assignment process is detailed in Algorithm~\ref{alg:rideassignment_pseu}. In line~2, an online deadhead control algorithm determines the deadhead limit to be used during the next batch assignment. This limit, in turn, defines the set of eligible drivers for each ride request in line~4.


\begin{algorithm}[!h]
\SetAlgoLined
$\mathcal{N}_b \leftarrow$ Set of ride requests in the assignment queue during assignment batch $b$\;
$d \leftarrow $ Upper-bound on the deadhead distance of trips in batch $b$ returned by online deadhead control algorithm\;
\For{each ride request $n \in \mathcal{N}_b$}
{
 $\mathcal{M}_{n,b} = \{m\ | \ m$ is available, deahdead distance between $m$ and  $n$ is at most $d \}$\;
 $m \leftarrow $ a driver in $\mathcal{M}_{n,b} $ with lowest trip emission\;
 Assign $m$ to $n$\;
}
 \caption{Online Ride Assignment in Batch $b$}
 \label{alg:rideassignment_pseu}
\end{algorithm}
The objective function in \probName comprises two terms: $E_{N}$, the time-averaged expected emission reduction, and $R_{N}$, the average rate of ride assignments. The first term addresses emissions, while the second ensures lower passenger wait times. To rigorously model two objective terms $E_{N}$ and $R_{N}$, we begin by introducing some additional notations as follows. Let $\textbf{B}$ denote the total number of batches required to process all ride requests, and let $n_b$ represent the number of requests assigned in batch $b$.

To model the time-averaged expected emissions reduction $E_N$, we define \( g_M(d) \) as the expected emissions saved by limiting the deadhead distance to \( d \) for a platform with \( M \) drivers, compared to always assigning the nearest driver. While emissions savings per assignment generally depend on time and location, under the assumption that ride requests are independent and drawn from a time-independent distribution over features (e.g., request time, pickup, and drop-off), the drivers' spatial distribution and the ratio of drivers to requests stabilize. This allows \( g_M(d) \) to be treated as a time- and location-independent function. We assume the existence of such a function, although its exact form is unknown to the algorithm, which must estimate it from prior (possibly noisy) observations.

To model the time-averaged expected emissions reduction $E_N$, we define \( g_M(d) \) as the expected emissions saved by limiting the deadhead distance to \( d \) for a platform with \( M \) drivers, compared to always assigning the nearest driver. While the emission reduction for each ride assignment generally depends on both time and location, under certain assumptions, these dependencies can be relaxed for \( g_M(d) \). Specifically, if ride requests are independent, with features such as request time, pickup, and drop-off locations drawn from a time-independent distribution, the resulting distribution of driver locations also becomes time-independent, and the ratio between the number of available drivers and ride requests in a particular region becomes time- and location-independent. This implies that \( g_M(d) \) can be treated as a time- and location-independent function. In this work, we assume the existence of such a function \( g_M(d) \), though the online algorithm \emph{does not} have prior knowledge of its exact form. Instead, it must approximate \( g_M(d) \) using estimation methods based on prior, potentially noisy, observations. Based on this, we introduce the first term of the objective of \probName, $E_{N}$, as 
\begin{equation}
    E_{N} = \frac{\mathbb{E}\bigg[\sum_{b = 1}^{\textbf{B}} \sum_{d \in \mathcal{D}} a_{b,d} \cdot n_b \cdot g_M(d)\bigg]}{\mathbb{E}[T_N]},
\end{equation}
where $T_N$ denotes the time from the first assignment to the drop-off of the last ride, $\mathcal{D} = \{d_{\min}, \dots, d_{\max}\}$ represents the set of potential deadhead distance limits available to the online algorithm, and $a_{b,d}$ is a binary decision variable indicating whether the deadhead distance in batch $b$ was limited to $d$. For simplicity, we normalize the function $g_M(d)$ to range between 0 and 1, with $g_M(d_{\min}) = 0$ and $g_M(d_{\max}) = 1$. The algorithm selects $d$ at the start of each batch, determining the maximum allowable deadhead distance for the upcoming batch.  
Last, it is worth noting that $E_N$ represents the emission reduction achieved by the online algorithm compared to a ridesharing platform that assigns requests to the nearest driver. Maximizing $E_N$ over the horizon $T_N$ is equivalent to minimizing the expected emissions of the ridesharing platform within that period. 

The second term in the objective function, $R_{N}$,  represents the global ride assignment rate. It seeks to minimize the total service time $T_N$  required to assign all ride requests. By reducing the total service time, the algorithm indirectly minimizes rider waiting times. The calculation of $R_{N}$ is given below.
\begin{equation}
    R_{N} = \frac{\mathbb{E}\bigg[\sum_{b = 1}^{\textbf{B}} \sum_{d \in \mathcal{D}} a_{b,d} \cdot n_b \bigg]}{\mathbb{E}[T_N]}.
\end{equation}
Combining these terms with the coefficient $\alpha$, the \probName is formulated as a maximization problem as follows.
\begin{subequations}
\label{eq:RideAssign}
    \begin{eqnarray}
    \label{eq:objective}
      [\probName] & \max &  \alpha \cdot E_{N} + (1-\alpha) \cdot R_{N} \\
      \label{eq:constraint_a}
      &\textrm{s.t.,} & \sum_{d\in \mathcal{D}}  a_{b,d} \leq 1  \qquad  \forall b, \\
       \label{eq:constraint_c}
      &\textrm{vars.,} & \quad \quad a_{b,d} \in \{0, 1\}  \qquad  \forall b, d,
    \end{eqnarray}
\end{subequations}
where parameter $0 \leq \alpha \leq 1$ represents the weight of the term $E_N$ 
in the objective value, determining the relative importance of maximizing the term related to reducing emissions ($E_N$) versus the term related to maximizing ride assignment rate ($R_N$). Using \probName to manage emissions and wait times simplifies the analysis of online algorithms while still enabling the minimization of both expected emissions and rider waiting times within the ridesharing system. 


During peak demand periods, long queues of unassigned ride requests can lead to significantly increased rider wait times~\cite{wang2018understanding,li2023understanding,lu2021efficiency}. To maintain responsiveness in such situations, the algorithm should prioritize reducing wait times by tightening deadhead distance limits and increasing assignment frequency, thereby placing greater emphasis on \(R_N\). Conversely, when the assignment queue is short, the algorithm can shift its focus toward environmental impact by matching riders with low-emission drivers, even if they are located farther away---placing greater emphasis on \(E_N\). This strategy helps lower overall emissions by utilizing more efficient vehicles when feasible. However, making such real-time adjustments remains challenging for online algorithms due to uncertainty in future ride requests, demand fluctuations, and variability in trip distances.



\section{\algName: An Online Algorithm for \probName}
\label{sec:LARA}

\probName can be modeled within the framework of Lyapunov optimization for renewal systems~\cite{neely2012dynamic}, where a decision maker selects actions sequentially, and the duration of each renewal frame depends on the chosen actions. In the context of the online deadhead control problem, each frame corresponds to the expected duration of trips in the upcoming batch, spanning from assignment to drop-off. We adopt the Lyapunov optimization approach—commonly used for designing online algorithms in renewal systems—to develop a robust and adaptive online algorithm for \probName.
\label{sec:online_algorithm}
The design and analysis of \algName rely on the following two assumptions:  
\begin{assumption}[Independence]
\label{asm:independent_req}
Ride requests are independent. The time and location of each new ride request are independent of previously posted requests.
\end{assumption}
\begin{assumption}[Continuity]
\label{asm:cont_time}
The rate of new ride requests is sufficiently high, such that at least one request is available for assignment in each batch.  
\end{assumption}  
We note that both assumptions are required for the theoretical analysis of \probName; however, the proposed algorithms remain executable even when the assumptions do not strictly hold. Moreover, these assumptions are often reasonable in practice. Assumption~\ref{asm:independent_req} is valid in real-world scenarios, as most riders request rides independently. Assumption~\ref{asm:cont_time} ensures continuity in the assignment process across all $\textbf{B}$ batches, maintaining a high-demand regime. While there may be low-demand periods—such as late-night hours—during which no requests occur, the system can still be analyzed over subintervals where the assumption holds.

\subsection{Design of \algName}
In the following, we introduce a \underline{L}yapunov-based \underline{A}lgorithm for \underline{R}ide \underline{A}ssignment (\algName), which provides a near-optimal solution to \probName. \algName dynamically adjusts the deadhead distance limit based on the number of unassigned ride requests without requiring any prior knowledge of future ride arrivals. 

To control rider waiting times, \algName employs a parameter, $Q_{\max}$, which sets the maximum allowable number of unassigned requests in the assignment queue. The aim is to keep the number of requests in the assignment queue below this threshold. Importantly, $Q_{\max}$ sets a worst-case upper bound on rider wait times, independent of the platform's expected average waiting time. Let $N_b$ represent the number of unassigned requests during assignment batch $b$. \algName computes $Q(b) = Q_{\max} - N_b$, which it then uses to set the deadhead distance limit for rides assigned in batch $b$.


When $Q(b)$ is small (indicating a large number of unassigned requests), the algorithm selects shorter deadhead distance limits to expedite the assignment process and reduce riders' waiting times. Conversely, when $Q(b)$ is large (indicating fewer unassigned requests), \algName increases the deadhead distance limit to optimize the ride assignments, aiming to reduce emissions. In each batch, \algName solves a single-batch maximization problem to determine the deadhead distance limit, which leads to a near-optimal solution for the entire \probName over $\textbf{B}$ batches.

As discussed in Section~\ref{sec:DCP}, online algorithms do not have perfect information about the function \(g_M\). \algName utilizes its estimation, \(\hat{g}_M\), to make deadhead control decisions. In each assignment batch \(b\), \algName first computes the number of requests in the assignment queue and evaluates \(Q(b)\) accordingly. It then determines the deadhead distance limit by solving the following optimization problem:
\begin{subequations}
\label{eq:OnlineAlg}
    \begin{eqnarray}
    \label{eq:alg_obj}
      \max_{d\in \mathcal{D}} &  \sum_{d} a_{b,d} \cdot \frac{Q_{\max}(1-\alpha + \alpha \cdot \hat{g}_M(d)) - Q(b)  }{d + d_{t,b}}, \\
      \label{eq:alg_sing_assign}
       \textrm{s.t.,} & \sum_{d\in \mathcal{D}}  a_{b,d} \leq 1  \quad \forall b, \\
       \label{eq:alg_vars}
      \textrm{vars.,} & \qquad  a_{b,d} \in \{0, 1\}.
    \end{eqnarray}
\end{subequations}
where $\{a_{b,d} \mid \forall b, d\}$ denotes the decision vector, and $d_{t,b}$ is the average trip distance of ride requests in batch \( b \). The constraint~\eqref{eq:alg_sing_assign} ensures that exactly one deadhead distance limit is selected for each batch. Intuitively, Equation~\eqref{eq:alg_obj} corresponds to the objective~\eqref{eq:objective} for a single batch and a single ride assignment (i.e., \( n_b = 1 \)). In this setting, the time horizon length is proportional to the sum of average deadhead and trip distances. Moreover, the numerator includes a shift of \( Q(b)/Q_{\max} \), which scales with the number of unassigned requests. This term enables \algName to adaptively adjust the deadhead limit based on \( Q(b) \), balancing emissions reduction with maximizing the assignment rate.

\section{Theoretical Analysis of \algName}
\label{sec:thm_analysis}
In this section, we provide a theoretical analysis of \algName's performance. We first derive a closed-form expression for the dynamics of \( Q(b) \) across different batches. Then, in Theorem~\ref{thm:perf_analysis}, we show that \algName achieves an objective value within a constant gap of the optimal solution for \probName, assuming perfect knowledge of \( g_M \) is given in advance. Finally, in Theorem~\ref{thm:perf_analysis_noisy}, we extend this result to analyze \algName's performance under no prior information about \( g_M \). 



\algName uses the value of $Q(b)$ to make decisions for batch $b$.
Dynamics of $Q(b)$ are directly influenced by the number of requests assigned during batch $b$ and the number of new requests posted before the next assignment batch. Let $r_b$ denote the number of ride requests posted between assignment batches $b$ and $b+1$. The update rule for $Q(b)$ is given by:
\begin{equation}
   \label{eq:q_update_rule}
    Q(b+1) = Q(b) - r_b + \sum_{d\in\mathcal{D}} a_{b,d}\cdot n_b,
\end{equation}
where $n_b$ represents the number of ride requests assigned during batch $b$ and is bounded by $n_b \leq \min (M_b, N_b)$ where $M_b$ denotes the number of available drivers during batch assignment $b$. The last term reflects whether the online algorithm has selected any deadhead distance limit for batch $b$ considering constraint~\eqref{eq:alg_sing_assign}, and if so, how many ride assignments occurred during that batch.  We leverage the given dynamic of \(Q(b)\) to establish a bound on the performance gap between \algName and the optimal offline algorithm, as formalized in the following theorem.



\begin{theorem}
\label{thm:perf_analysis}
Let \( OBJ \) and \( OBJ_o \) denote the objective values of \algName and the optimal offline algorithm for \probName, respectively. Under perfect estimation of the function \( g_M(d) \), i.e., \( \hat{g}_M(d) = g_M(d) \), and considering an asymptotically long time horizon, i.e., \( \textbf{B} \to \infty \), the following bound holds:
\begin{equation}
    OBJ_o - \frac{\E[r_b^2]}{Q_{\max} \cdot \E[t_n]} \leq OBJ,
\end{equation}
where \( \E[t_n] \) denotes the unconditional expected duration of a trip, from assignment to drop-off, under \algName's deadhead control.

\begin{proof}
We leverage the fact that the set of possible actions at any time remains unchanged, and apply Lemma~1 from~\cite{neely2012dynamic}, which establishes that for online renewal systems, there exists a stationary algorithm—one that selects a fixed action in every interval—that achieves the optimal objective value. Since \probName falls within the class of such systems, this result supports our analysis of \algName's performance. Specifically, the stationary algorithm in the context of \probName corresponds to choosing a fixed deadhead distance limit across all assignment batches. 


While identifying the optimal stationary algorithm is infeasible in practice due to uncertainty in future ride requests, its existence, and its equivalence in performance to the optimal offline algorithm, makes it a valuable benchmark. Therefore, we use the best stationary policy as a reference point to evaluate the performance of the proposed online deadhead control algorithm, \algName.


We begin by defining a Lyapunov function $L(b)$ and the conditional Lyapunov drift function, $D(b)$, as follows: 
\begin{align*}
    L(b) =& \frac{1}{2} Q(b)^2,\\
    D(b) =& \E\left[L(b+1) - L(b) |Q(b), M_b\right].
\end{align*}
Leveraging the update rule for $Q(b)$ from Equation~\eqref{eq:q_update_rule}, we get:
\begin{align*}
        D(b) =& \frac{1}{2} \E[(r_b - \sum_{d\in \mathcal{D}} a_{b,d}\cdot n_b)^2| Q(b), M_b]\\
    -& Q(b) \E[(r_b - \sum_{d\in \mathcal{D}} a_{b,d}\cdot n_b) | Q(b), M_b].
\end{align*}
By subtracting a value of $Q_{\max}\cdot n_b\sum_{d\in \mathcal{D}} a_{b,d} \cdot (1-\alpha + \alpha \cdot g_M(d))$ from both sides, we obtain:
\begin{align}
    D(b)  -& Q_{\max}\cdot n_b\sum_{d\in \mathcal{D}} a_{b,d} \cdot (1-\alpha + \alpha \cdot g_M(d))\notag\\
    \leq & \frac{1}{2} \E[(r_b - \sum_{d\in \mathcal{D}} a_{b,d}\cdot n_b)^2| Q(b), M_b]\notag\\
    -& Q(b) \E[(r_b -  \psi\sum_{d\in \mathcal{D}} a^{*}_{b,d}\cdot n_b) | Q(b), M_b] \notag\\
    -& \psi Q_{\max}\cdot n_b\sum_{d\in \mathcal{D}} a^{*}_{b,d} \cdot (1-\alpha + \alpha \cdot g_M(d))\notag\\
    \leq & \frac{1}{2} \E[(r_b - \sum_{d\in \mathcal{D}} a_{b,d}\cdot n_b)^2| Q(b), M_b]\notag\\
    +& Q(b) \E[\sum_{d} a_{b,d} (d + d_{t,b})] \E[(\sum_{d\in \mathcal{D}} a^{*}_{b,d} \frac{ n_b}{(d + d_{t,b})} \notag\\
    -& \frac{r_b}{\sum_{d} a_{b,d} (d + d_{t,b})} ) | Q(b), M_b]\notag\\
    \label{eq:proof_perfect_br1}
    -& \psi Q_{\max}\cdot n_b\sum_{d\in \mathcal{D}} a^{*}_{b,d} \cdot (1-\alpha + \alpha \cdot g_M(d)),
\end{align}
where $a^{*}_{b,d}$ is the action of optimal stationary algorithm for \probName, and $\psi = \frac{\E[\sum_{d} a_{b,d} (d + d_{t,b})]}{\E[\sum_{d} a^{*}_{b,d} (d + d_{t,b})]} $. The above inequality holds since \algName's action results from solving the maximization problem~\eqref{eq:alg_obj}, i.e.
\begin{align*}
    &\sum_{d} a^{*}_{b,d} \frac{\left(Q_{\max}\cdot(1-\alpha + \alpha \cdot g_M(d)) - Q(b) \right) }{d + d_{t,b}} \\
    \leq& \sum_{d} a_{b,d}  \frac{ \left(Q_{\max}\cdot(1-\alpha + \alpha \cdot g_M(d)) - Q(b) \right) }{d + d_{t,b}}.
\end{align*}

The second term is always negative, as the expected rate of assignment under the optimal stationary algorithm cannot exceed the expected posting rate of new ride requests. By taking the conditional expectation of both sides, summing over all batches, and dividing by $Q_{\max}\cdot \textbf{B} \cdot \E[d+d_{t,b}]$, we get:
\begin{align*}
    & \frac{1}{Q_{\max}\cdot \textbf{B}\cdot   \E[d+d_{t,b}]}  \sum_{b=1}^{\textbf{B}} D(b) - \sum_{b=1}^{\textbf{B}} n_b\sum_{d\in \mathcal{D}} a_{b,d} \frac{ 1-\alpha + \alpha \cdot g_M(d)}{\textbf{B} \cdot   \E[d+d_{t,b}]}\\
    \leq &\frac{1}{2Q_{\max}\cdot \textbf{B} \cdot   \E[d+d_{t,b}]} \sum_{b=1}^{\textbf{B}} \E[(r_b - \sum_{d\in \mathcal{D}} a_{b,d}\cdot n_b)^2| Q(b), M_b]\\
    -&   \sum_{b=1}^{\textbf{B}} n_b \sum_{d\in \mathcal{D}} a^{*}_{b,d} \frac{ 1-\alpha + \alpha \cdot g_M(d)}{\textbf{B} \cdot   \E[d+d_{t,b}]}\\
    \leq &\frac{\E[r_b^2]}{Q_{\max}\cdot \E[d+d_{t,b}]} -\sum_{b=1}^{\textbf{B}} n_b \sum_{d\in \mathcal{D}} a^{*}_{b,d} \frac{ 1-\alpha + \alpha \cdot g_M(d)}{\textbf{B} \cdot   \E[d+d_{t,b}]}.
\end{align*}
Under Assumption~\ref{asm:cont_time} and taking the limit as $\textbf{B} \rightarrow \infty$ completes the proof.
\end{proof}
\end{theorem}
\begin{remark}
\label{rem:conflict_worstCaseWaiting}
Theorem~\ref{thm:perf_analysis} demonstrates that increasing the value of the parameter $Q_{\max}$ reduces the gap between the performance of \algName and the optimal algorithm. This highlights a trade-off between minimizing expected emissions and managing the maximum rider waiting time within the ridesharing system. In addition, \algName's performance converges that of optimal offline algorithm if high values for $Q_{\max}$ is used which shows the optimality of \algName in such cases.
\end{remark}
Although Theorem~\ref{thm:perf_analysis} provides key insights into the performance of \algName, in practice, \algName relies on predicted values of \( g_M(d) \), which are inherently imperfect. The following theorem extends the result of Theorem~\ref{thm:perf_analysis} to account for noisy predictions.

\begin{figure*}[t!]
\centering
\begin{tabular}{cc}
\multicolumn{2}{c}{\includegraphics[width=0.75\linewidth]{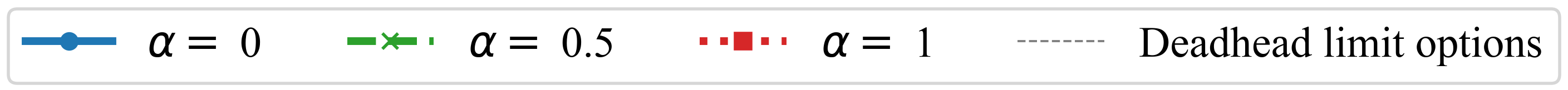}} \vspace{-1mm} \\
\includegraphics[width=0.48\linewidth]{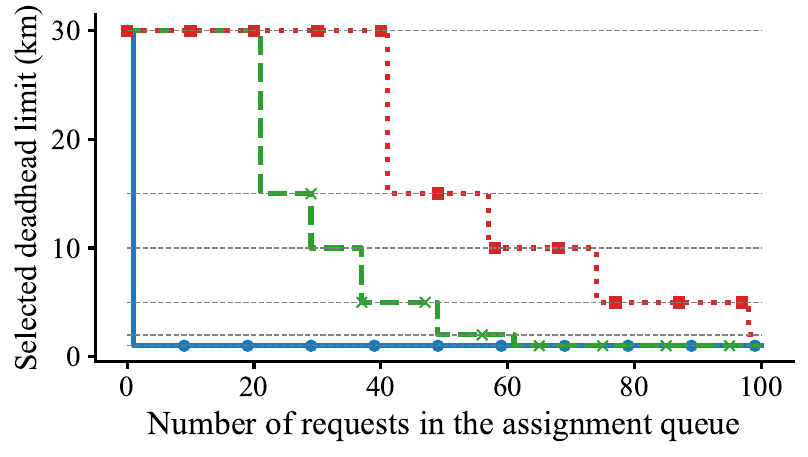} \hspace{10mm}&  \includegraphics[width=0.48\linewidth]{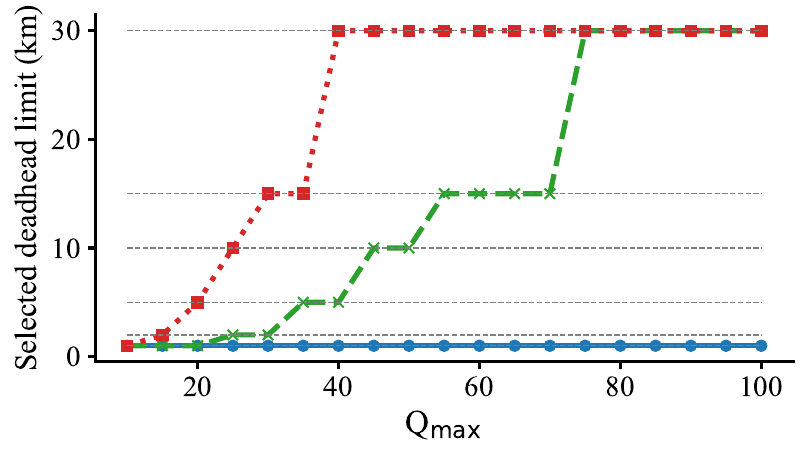} \\  
\end{tabular}
    \caption{(a) Deadhead distance limits selected by \algName as a function of the number of requests in the assignment queue with $Q_{\max}$ fixed at 100 (a), and as a function of $Q_{\max}$ with the assignment queue length fixed at 15 for different values of $\alpha$.}
    \label{fig:understanding}
\end{figure*}

\begin{theorem}
\label{thm:perf_analysis_noisy}
The result of Theorem~\ref{thm:perf_analysis} holds if \( \hat{g}_M(d) \) is derived from an unbiased estimator of \( g_M(d) \), satisfying:
\begin{equation} 
\lim_{b \to \infty} \E \left[|\hat{g}_{M,b}(d) - g_M(d)|\right] = 0, \quad \forall d,
\end{equation}
where \( \hat{g}_{M,b}(d) \) denotes the estimate of \( g_M(d) \) after batch \( b \).
\end{theorem}
\begin{proof}
%
Let $\hat{g}_M={M,b}(d)$ denotes the estimation of $g_M(d)$ until batch assignment $b$. Since $\hat{g}_M(d)$ is an unbiased estimation of function $g_M(d)$, there exists a batch index $b_0$ and positive constants  $C_0$ and $\gamma$ such that:
\begin{equation*}
    ||g_M(d) - \hat{g}_{M,b}(d) || \leq C_0 b_0^{-\gamma} \quad \forall d, b_0 < b \leq \textbf{B}.
\end{equation*}
Let define $\epsilon := C_0 b_0^{-\gamma}$, the above inequality yields that for any batch assignment $b > b_0$:
\begin{align*}
    &\sum_{d} a_{b,d} \frac{\left(Q_{\max}\cdot(1-\alpha + \alpha \cdot \hat{g}_{M,b}(d)) - Q(b) \right) }{d + d_{t,b}} \\
    \geq& \sum_{d} a^*_{b,d}  \frac{ \left(Q_{\max}\cdot(1-\alpha + \alpha \cdot \hat{g}_{M,b}(d)) - Q(b) \right) }{d + d_{t,b}} \\
    \geq& \sum_{d} a^*_{b,d}  \frac{ \left(Q_{\max}\cdot(1-\alpha + \alpha \cdot g_{M}(d)) - Q(b) \right) }{d + d_{t,b}} - \frac{\alpha \cdot Q_{\max} \cdot \epsilon}{d_{\min}}.
\end{align*}
Now following the Equation~\eqref{eq:proof_perfect_br1}, we get:
\begin{align*}
    \Rightarrow& \frac{1}{Q_{\max}\cdot \textbf{B}\cdot   \E[d+d_{t,b}]}  \sum_{b=1}^{\textbf{B}} D(b) - \sum_{b=1}^{\textbf{B}} n_b\sum_{d\in \mathcal{D}} a_{b,d} \frac{ 1-\alpha + \alpha \cdot \hat{g}_M(d)}{\textbf{B} \cdot   \E[d+d_{t,b}]}\\
    \leq &\frac{\E[r_b^2]}{Q_{\max}\cdot \E[d+d_{t,b}]} -\sum_{b=1}^{\textbf{B}} n_b \sum_{d\in \mathcal{D}} a^{*}_{b,d} \frac{ 1-\alpha + \alpha \cdot g_M(d)}{\textbf{B} \cdot   \E[d+d_{t,b}]}\\
    +& C_1 \cdot b_0 + \sum_{b=b_0}^{\textbf{B}} \frac{\alpha \cdot Q_{\max} \cdot \epsilon}{d_{\min}},
\end{align*}
where $C_1$ is a constant. Setting $b_0 = \sqrt{\textbf{B}}$ and taking the limit as $\textbf{B} \rightarrow \infty$ completes the proof.
\end{proof}

Theorem~\ref{thm:perf_analysis_noisy} establishes that if an unbiased estimator of \( g_M(d) \) is used in \algName, the lower bound on its objective value remains consistent with the result of Theorem~\ref{thm:perf_analysis}. Since \( g_M(d) \) is a time-independent function, obtaining an unbiased estimate for it is a well-studied problem. For discretized deadhead distance values in the set \( \mathcal{D} \), estimating \( g_M(d) \) reduces to solving \( |\mathcal{D}| \) independent estimation tasks. Several well-known estimation techniques can be applied, including Monte Carlo estimation~\cite{dagum2000optimal,sadowsky1993optimality}, Holt-Winters smoothing~\cite{chatfield1978holt}, and the Kalman filter with sequential updates~\cite{welch1995introduction,jover1986parallel}. Each of these methods provides an unbiased estimate of \( g_M(d) \) for any \( d \in \mathcal{D} \). Over a sufficiently long time horizon, as the number of ride requests increases substantially, the estimates produced by these methods converge to the true values of \( g_M(d) \).


\textbf{Scalability and time complexity of \algName.} Beyond optimizing the objective, the computational complexity of an online algorithm is a crucial factor. An algorithm requiring a high number of operations per iteration faces scalability issues, making it impractical for real-world use as the problem size—i.e., the number of requests and drivers in \probName—grows. \algName solves the optimization problem~\eqref{eq:alg_obj} in each batch with a complexity of $\mathcal{O}(|\mathcal{D}|)$. Since the size of the deadhead distance options, $\mathcal{D}$, is independent of the number of requests and drivers, the complexity of \algName’s decision-making remains constant, ensuring its scalability and efficiency.

\section{Experimental Analysis}
\label{sec:experiments}
We begin this section with a high-level analysis of \algName's behavior. Specifically, we examine how the parameters $Q_{\max}$ and $\alpha$ influence its decisions. We then conduct an extensive evaluation of \algName using both synthetic data and real-world data from the RideAustin dataset~\cite{rideaustin-dataset}. This dual approach allows us to compare \algName against alternative algorithms in both controlled experimental settings and realistic ridesharing environments. Synthetic datasets enable the simulation of high-demand scenarios and facilitate analysis of key factors such as trip distances, driver availability, and batch assignment intervals. Although publicly available ridesharing datasets are limited and often lack critical information (e.g., pickup/drop-off locations or vehicle attributes), the RideAustin dataset remains the most suitable for evaluating \algName. While it does not fully reflect high-demand conditions, it provides valuable insights into the algorithm’s real-world performance. Finally, our ride-assignment process follows a batch system, where the algorithm assigns unallocated requests to available drivers in successive batches. Unassigned requests are carried over to the next cycle for reassignment.

\begin{figure*}[t!]
	\centering
	\subfigure{\includegraphics[width=0.46\textwidth]{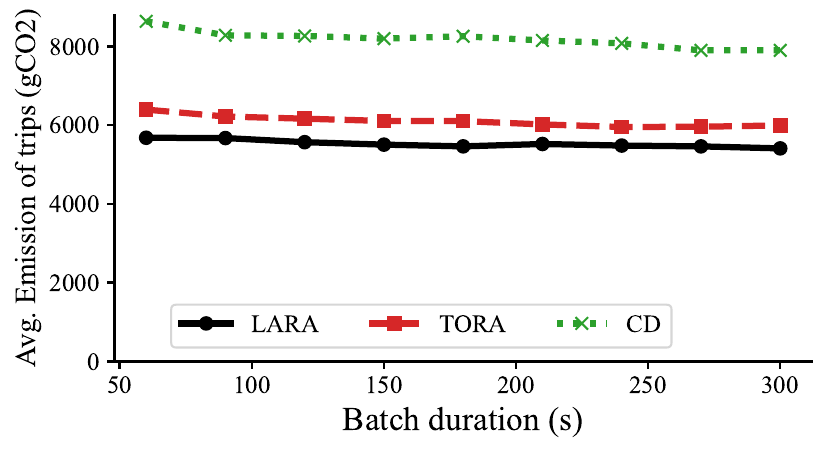}} \hspace{5mm}
	\subfigure{\includegraphics[width=0.46\textwidth]{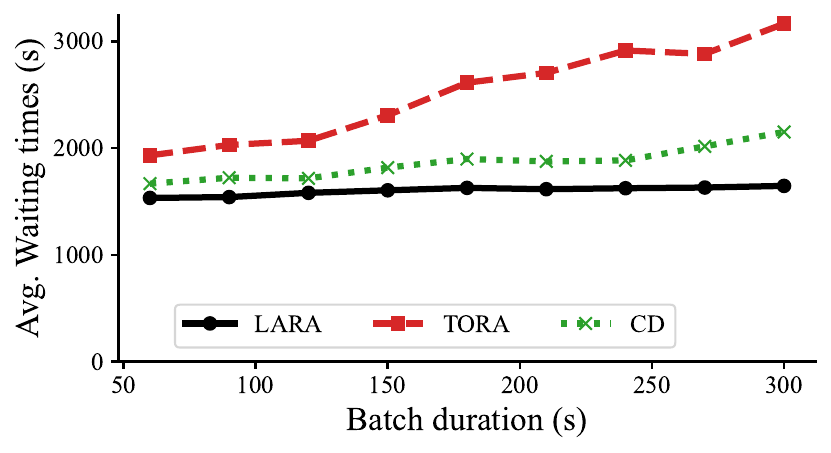}}
	\caption{Average trip emissions (left) and rider waiting times (right) as a function of batch duration for \algName and comparison algorithms. Increasing the batch duration generally leads to lower average emissions but higher rider waiting times.}
	\label{fig:synthetic_batchDur}
\end{figure*}

\subsection{Understanding of the Behavior of \algName}
First, we aim to scrutinize the deadhead control process within \algName. \algName dynamically selects the deadhead distance limit based on factors such as the number of requests in the assignment queue and hyper-parameters like $Q_{\max}$. To analyze the impact of these factors on \algName's decision-making, we conduct a simple analysis using a version of \algName with $Q_{\max} = 100$ and deadhead distance limit options of 1, 2, 5, 10, 15, and 30 km. We assume there are 15 ride requests in the assignment queue and evaluate how \algName selects deadhead limits by varying the number of requests and the value of $Q_{\max}$ in two separate analyses. In this analysis, we use a logarithmic form for the function $g_M(d)$, where $g_M(d) = C_0 \cdot \log(d)$, with $C_0 = 1/\log(d_{\max}) \approx 0.20$ as the normalization constant.

Figure~\ref{fig:understanding}(a) illustrates the influence of the number of requests in the assignment queue on the deadhead distance limit selected by \algName for different values of $\alpha$. When the queue has relatively few ride requests, \algName selects higher deadhead distance limits to prioritize reducing emissions. However, as the number of requests increases, \algName shifts focus towards reducing rider waiting times by selecting lower deadhead limits. Additionally, increasing the parameter $\alpha$ causes \algName to emphasize more on the emission term in the objective function of Equation~\eqref{eq:objective}, leading to the selection of higher deadhead distance limits. In Figure~\ref{fig:understanding}(b), we plot the deadhead limits selected by \algName as a function of $Q_{\max}$ for different values of $\alpha$. When the number of requests in the assignment queue is fixed (15 in this analysis), increasing $Q_{\max}$ widens the gap between the assignment queue length and $Q_{\max}$. This results in \algName selecting higher deadhead distance limits as the system allows for a larger buffer to focus on emission reduction.

\noindent 
\emph{\textbf{Key takeaway.}  As the number of unassigned ride requests increases, \algName lowers the deadhead distance limits to prioritize faster ride assignments. Higher values of $Q_{\max}$ leads \algName to select larger deadhead distance limits, providing greater flexibility to reduce emissions.}


\textbf{Choosing parameters $\alpha$ and $Q_{\max}$.} In practice, the parameter $\alpha$ is chosen by stakeholders and decision makers based on their relative priorities between reducing emissions and minimizing waiting times. As shown in Figure~\ref{fig:understanding}, varying $\alpha$ influences the deadhead distance limits selected by \algName. The parameter $Q_{\max}$, on the other hand, reflects the platform’s tolerance for maximum allowable waiting times. Specifically, doubling $Q_{\max}$ can result in a doubling of the worst-case waiting time under adverse conditions. Selecting an appropriate value for $Q_{\max}$ should take into account the operational and statistical characteristics of the ridesharing platform, such as demand patterns, and traffic conditions.



\subsection{Experiments on Synthetic Data} \textbf{Comparison Algorithms.} We compare the performance of \algName against two algorithms: (1) \TORA~\cite{sahebdel2024holistic}, the only existing emission-aware ride assignment algorithm which also considers deadhead emissions; and (2) a heuristic algorithm that assigns riders to the closest available driver, referred to as \cd. \TORA is considered as a baseline algorithm for emission reduction while \cd is a baseline for minimizing passengers' waiting times in a greedy manner.


\noindent \textbf{Experimental Setup.} We generated traces consisting of 50,000 ride requests, with a new request posted on the ridesharing platform every 5 seconds. Pickup and drop-off locations are selected within area of Austin, Texas, ensuring an average trip distance of 15 km with a standard deviation of 5 km. The assignment happens every 2 minutes and the platform operates with 500 drivers, resulting in a per hour driver-to-request ratio of approximately $60\%$, aligning with observed Uber/Lyft statistics~\cite{uber-lyft-data1}. Additionally, vehicle emissions range between 70–300 gCO2/km, reflecting emission levels found in the RideAustin dataset and consistent with previous studies~\cite{sahebdel2023data, sahebdel2024holistic}.

For \algName, we used deadhead distance limits options of 1, 2, 5, 10, 15, and 30 km, with $\alpha = 0.75$ and $Q_{\max} = 240$, which encourages \algName to keep the portion of waiting time spent in the assignment queue to under 20 minutes. This is achieved as 24 new ride requests are generated in each batch, and 10 batch assignments occur within a 20-minute window. Additionally, we applied the Monte Carlo estimation method to estimate the function $g_M$. Specifically, when \algName assigns a ride while limiting the deadhead distance to $d$, it calculates the emission reduction compared to the closest driver’s emission and updates the estimate of $g_M(d)$ based on Monte Carlo estimation strategy. In each evaluation, all parameters are kept constant except for one, which is varied to assess its impact. For each test, we report the average trip emissions and rider waiting times for \algName and the comparison algorithms.

\begin{figure*}[t]
	\centering
	\subfigure{\includegraphics[width=0.46\textwidth]{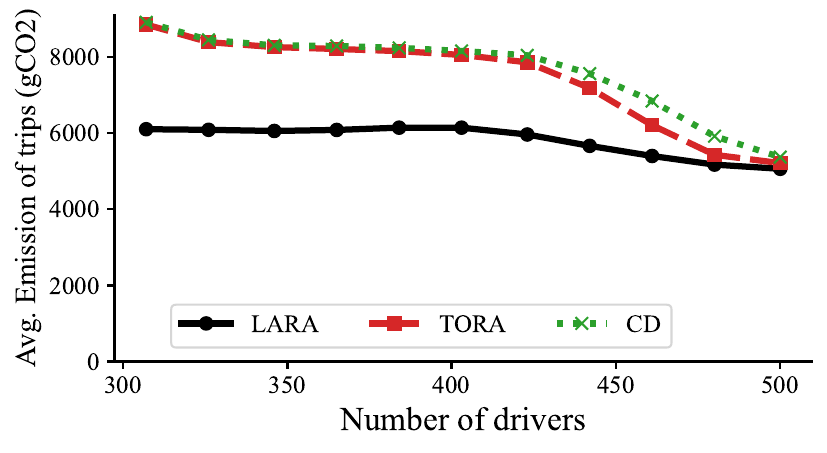}} \hspace{5mm}
	\subfigure{\includegraphics[width=0.46\textwidth]{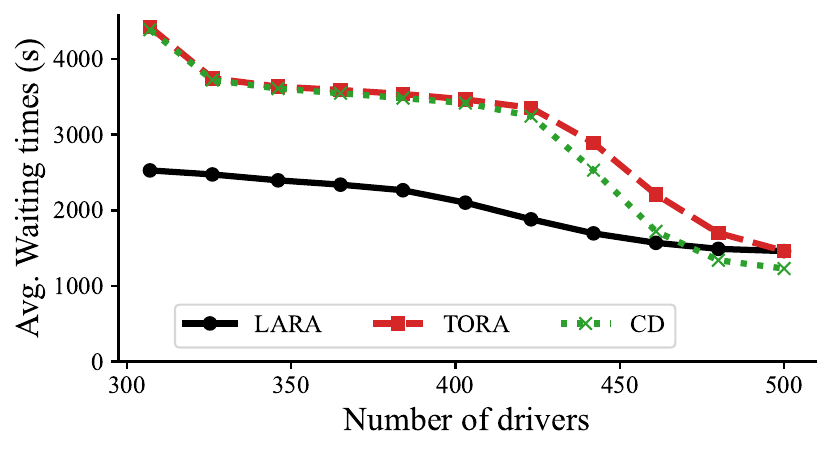}}
	\caption{Average trip emissions (left) and rider waiting times (right) as a function of the number of drivers for \algName and comparison algorithms. Fewer drivers simulate high-demand hours, leading to a wider performance gap between \algName and the other algorithms.}
	\label{fig:synthetic_driver}
\end{figure*}

\noindent \textbf{Experiment Results.} Figure~\ref{fig:synthetic_batchDur} illustrates the average trip emissions and rider waiting times as a function of batch duration for \algName and the comparison algorithms. As seen in the figure, \algName consistently achieved the lowest average emissions and waiting times across all batch durations. Specifically, compared to \cd, \algName reduced the average trip emissions by $30\%-34\%$, while this range for \TORA was $24.1\%-26.4\%$. In terms of waiting times, \algName resulted in average waiting times between $1500$ and $1600$ seconds, whereas for \cd, the range was $1600$ to $2200$ seconds, and for \TORA, between $1900$ and $3200$ seconds. During high-demand hours, the number of available drivers is low, leaving algorithms like \cd and \TORA with less flexibility to reduce emissions, as they tend to assign riders to any available driver. In contrast, \algName prioritizes finding rides with lower emissions and shorter deadhead distances, leading to lower average emissions and waiting times overall.

\noindent
\emph{\textbf{Key takeaway.}  Increasing the batch assignment duration reduces the average trip emissions but increases rider waiting times. This effect is more pronounced for \cd and \TORA than for \algName. }

\begin{figure*}[t]
	\centering
	\subfigure{\includegraphics[width=0.46\textwidth]{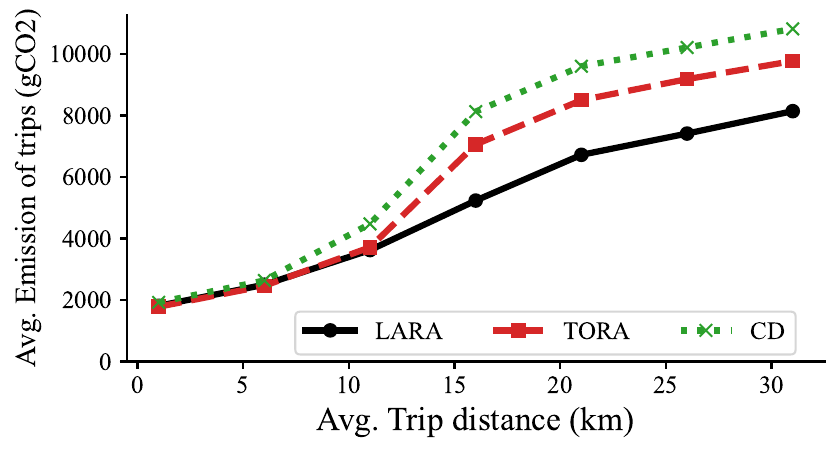}} \hspace{5mm}
	\subfigure{\includegraphics[width=0.46\textwidth]{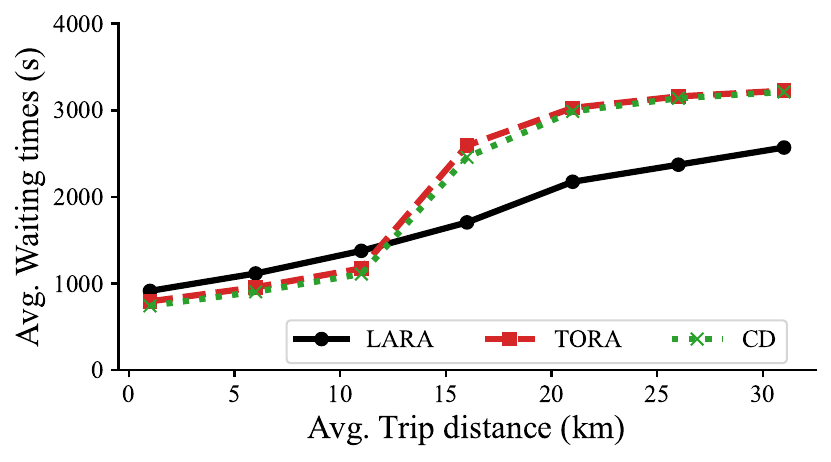}}
	\caption{Average trip emissions (left) and rider waiting times (right) as a function of trip distance for \algName and comparison algorithms. Longer trip distances increase the time each driver is allocated to a single trip, simulating high-demand hours more closely and widening the performance gap between \algName and the other algorithms.}
	\label{fig:synthetic_tripDistance}
\end{figure*}

In the second evaluation, we examine how the number of drivers in the ridesharing platform affects the performance of online ride assignment algorithms. Figure~\ref{fig:synthetic_driver} illustrates the average trip emissions and rider waiting times as a function of the number of drivers for \algName and comparison algorithms. Our results show a significant performance gap between \algName and the comparison algorithms when the number of drivers is low. In such scenarios, demand exceeds driver availability, which amplifies the performance difference between the optimal algorithm and alternatives like \TORA and \cd. For instance, with 310 drivers, \algName achieves $31\%$ lower emissions compared to \TORA and reduces waiting times by $43\%$. On the other hand, as the number of drivers increases, the performance gap between the algorithms narrows.  With larger driver pools (e.g., more than 480 drivers in our tests), the platform is no longer in a high-demand state, allowing algorithms like \cd to achieve the shortest waiting times. Nonetheless, even in these conditions, \algName continues to deliver the lowest average emissions.

\noindent
\emph{\textbf{Key takeaway.}   A low number of drivers in the synthetic dataset closely mirrors high-demand hours, during which the performance gap between \algName and the comparison algorithms is most pronounced. As the number of drivers increases, the gap narrows. In these conditions, \algName continues to achieve the lowest average emissions, while \cd delivers the shortest rider waiting times. }

In the final evaluation using the synthetic dataset, we assess how average trip distances affect the performance of ride assignment algorithms. Figure~\ref{fig:synthetic_tripDistance} illustrates the average trip emissions and rider waiting times as a function of average trip distance for various algorithms. As shown, when the average trip distance increases, the emissions and waiting times also rise across all algorithms. However, this increase is less pronounced for \algName, resulting in a wider performance gap compared to the other algorithms. The intuition behind this is that longer trips require drivers to spend more time on each ride, limiting the number of rides they can serve in a given period. Consequently, longer trip distances better simulate high-demand conditions, further highlighting the performance difference between \algName and the comparison algorithms. For example, with an average trip distance of 15 km, \algName achieves $25.8\%$ lower emissions and $33.7\%$ shorter waiting times compared to \TORA. Conversely, when trip distances are relatively short (e.g., less than 12 km in our tests), the platform is no longer in high-demand conditions. In these cases, \cd offers the shortest waiting times, while \TORA provides the lowest carbon emissions.

\noindent
\emph{\textbf{Key takeaway.} The average length of trips significantly influences high-demand conditions and the performance of online algorithms. Longer trips reduce the platform's capacity to quickly serve ride requests, widening the performance gap between heuristic algorithms like \TORA and \cd and near-optimal algorithms such as \algName.}

\begin{visionbox}{}
\noindent
\emph{\textbf{Overall takeaway.} Various factors, such as the number of drivers, ride request rates, batch assignment frequency, and trip distances, can lead to high-demand conditions. Under such conditions, \algName outperforms \cd and \TORA in reducing average emissions and rider waiting times. }
\end{visionbox}

\begin{figure*}[t]
\centering
\begin{tabular}{cc}
\multicolumn{2}{c}{\includegraphics[width=0.75\linewidth]{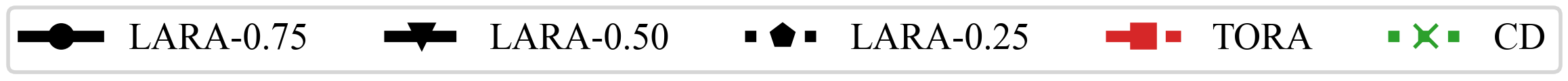}}\vspace{-1mm}\\
\includegraphics[width=0.44\linewidth]{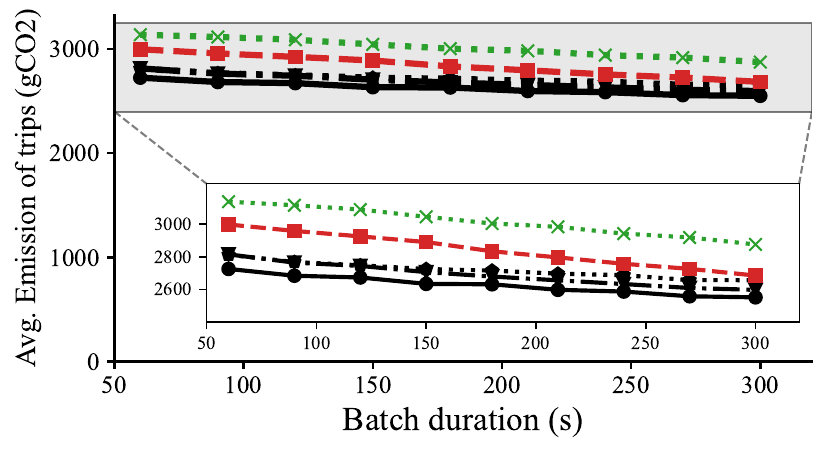} \hspace{5mm} &  \includegraphics[width=0.44\linewidth]{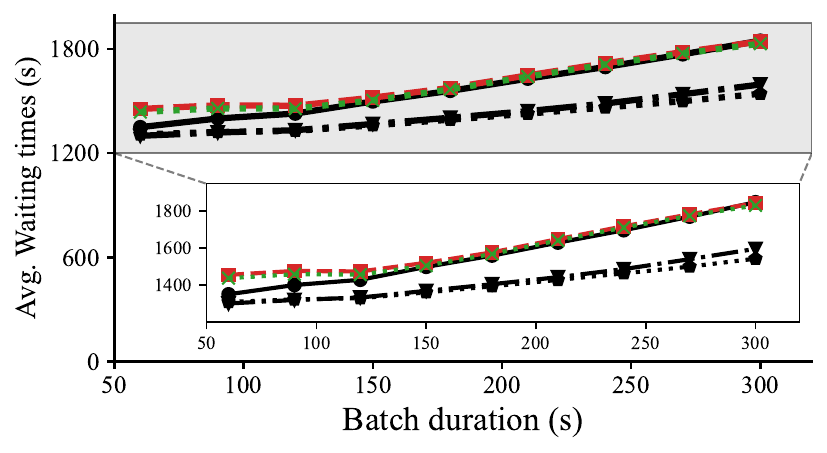}    
\end{tabular}
    \caption{Average trip emissions (left) and rider waiting times (right) as a function of batch duration for \algName and comparison algorithms. Higher values of $\alpha$ in \algName result in lower average emissions but come at the cost of increased rider waiting times.}
    \label{fig:dataset_batchDur} 
\end{figure*}

\subsection{Experiments on Real Dataset}
\label{sec:real-exp}
In this section, we evaluate the performance of \algName and baseline algorithms using the RideAustin dataset~\cite{rideaustin-dataset}. This dataset enables realistic simulations of ridesharing scenarios, allowing us to assess both average waiting times and emissions under various conditions. We use ride request data from the first week of December 2016, which includes 29,850 ride requests. To simulate high-demand conditions, we limit the number of available drivers so that the driver-to-request ratio is approximately 40\%—compared to the typical 50\% of real world conditions~\cite{uber-lyft-data1}. As in the synthetic experiments, we test deadhead distance limits of 1, 2, 5, 10, 15, and 30 km, set $Q_{\max} = 240$, and evaluate three versions of \algName with $\alpha$ values of 0.25, 0.5, and 0.75.

We report the average emissions per trip and the average rider waiting times for various batch durations. The results are shown in Figure~\ref{fig:dataset_batchDur}. As illustrated, increasing the batch duration leads to lower average emissions but longer waiting times. This trend is consistent with the results from the synthetic dataset (Figure~\ref{fig:synthetic_batchDur}). Across all batch durations, the three versions of \algName outperformed the comparison algorithms, achieving lower average emissions and shorter waiting times. Specifically, compared to \cd, \algName with $\alpha = 0.75$, $\alpha = 0.5$, and $\alpha = 0.25$ reduced emissions up to $11.1\%$, $11.3\%$, and $13.9\%$, respectively, while \TORA reduced emissions up to $6.6\%$. Additionally, \algName with $\alpha=0.25$ achieved between $8.7\%$ and $15.7\%$ shorter waiting times compared to \cd across the different batch durations.
\begin{table}[!t]
    \centering
    \caption{Percentage of rides assigned to drivers with low and high emission vehicles for \algName, \TORA, and \cd.}
    \begin{tabular}{lccccc}
               & \algName-0.75 & \algName-0.5 & \algName-0.25 & \TORA& \cd \\
        \hline \hline
        Low\ \ emission vehicles &  12.7  & 15.9  & 18.3 &  12.9  & 19.4   \\
        \hline
        High emission vehicles & 15.9  & 16.5  & 20.7 &  18.9  & 18.9   \\
        \hline
    \end{tabular}
    \label{tab:ridePerc}
\end{table}
\begin{table}[!t]
    \centering
    \caption{Percentage of distances traveled by drivers as deadhead miles for low, and high emission vehicles.}
    \begin{tabular}{lccccc}
               & \algName-0.75 & \algName-0.5 & \algName-0.25 & \TORA& \cd \\
        \hline \hline
        Low\ \ emission vehicles &  52.8  & 51.7  & 49.2 &  55.4  & 51.8   \\
        \hline
        High emission vehicles & 51.9  & 50.5  & 48.1 &  53.8  & 51.9   \\
        \hline
    \end{tabular}
    \label{tab:deadhead_to_trip}
\end{table}

\noindent
\emph{\textbf{Key takeaway.} Similar to the synthetic dataset results, increasing the batch assignment duration reduces average trip emissions but increases rider waiting times. Moreover, higher $\alpha$ values in \algName lead to lower emissions but longer waiting times.}

Finally, we report the percentage of ride requests performed by drivers with relatively low, and high emission vehicles to assess the performance of \algName and other comparison algorithms in providing fair ride assignment across different vehicles. To this end, we categorize vehicles with unit emission of less than $150 gCO2$ as a low emission vehicles and vehicles with unit emission of higher than $250 gCO2$ a high emission vehicles. Using this categorization approach, $24\%$ of dataset drivers categorized to low emission vehicles and $26\%$ categorized to high emission vehicles. The fraction of ride requests performed by low/high emission vehicles must be very similar to the fraction of those vehicles in the platform under control of fully fair algorithm.

In Table~\ref{tab:ridePerc}, and Table~\ref{tab:deadhead_to_trip} we present the percentage of ride requests fulfilled by drivers of low- and high-emission vehicles, along with the percentage of their total distance spent as deadhead miles. The results indicate that algorithms like \algName with high $\alpha$ values, or \TORA, assign fewer rides to low-emission vehicles. This occurs because these algorithms prioritize minimizing deadhead distances over reducing  total emissions of trips, causing low-emission vehicles to spend more of their total distance on deadhead miles. For example, under ride assignment of \TORA, $55.4\%$ of the distance traveled by low-emission vehicles was due to deadhead miles, compared to $49.2\%$ under \algName with $\alpha=0.25$ and $51.8\%$ under \cd. During high-demand periods, when all drivers are consistently busy, low-emission vehicles tend to travel more deadhead miles, resulting in fewer ride assignments. This underscores the trade-off between fair ride allocation and optimizing carbon reduction in ridesharing platforms.

\noindent
\emph{\textbf{Key takeaway.} The ride assignment of \algName indicates that reducing emissions requires a trade-off in equity among drivers, as those with low-emission vehicles typically handle rides with longer deadhead distances, resulting in a lower number of ride requests during peak demand hours.}

\section{Related Work}
\label{sec:related_work}
Recent advances in ridesharing systems have been fueled by a combination of theoretical modeling, empirical analysis, and data-driven techniques. From a theoretical standpoint, prior work has focused on modeling platform dynamics and designing mechanisms to improve operational efficiency. For example,~\citet{sadowsky2017impact} study the causal effect of ridesharing platforms on public transit ridership using a regression discontinuity approach. \citet{afeche2022ride} develop a game-theoretic model to characterize equilibrium behavior in ridesharing platforms, while~\citet{hu2022surge} and~\citet{castillo2022matching} investigate surge pricing mechanisms that dynamically adjust fares and wages to balance supply and demand, thereby reducing inefficient driver movement. Other efforts propose incentive-compatible payment schemes under non-stationary conditions~\cite{garg2022driver}, and spatio-temporal pricing strategies that promote spatial and temporal continuity in driver engagement~\cite{ma2020spatio}.

Environmental sustainability in ridesharing has also garnered significant attention. Several empirical studies quantify the potential emission reduction enabled by shared mobility. Jalali et al.\cite{jalali2017investigating} show that ridesharing platforms can reduce transportation emissions by up to 24\% through increased vehicle utilization. A complementary analysis using historical data from Tokyo reveals that up to 27\% of vehicle kilometers traveled could be saved through shared rides, resulting in emissions reductions of up to 84\%. More recent work has focused on algorithmic approaches to minimize both in-trip and deadhead emissions\cite{li2024shared,sahebdel2025lead,sahebdel2023data,sahebdel2024holistic}, while jointly optimizing service-level objectives such as rider waiting time and equitable assignment across vehicles with heterogeneous emissions profiles. These studies highlight critical trade-offs between efficiency, fairness, and sustainability.

A parallel line of work addresses matching and vehicle allocation under uncertainty. \citet{feng2021two} propose a two-stage stochastic model to manage uncertain rider and driver availability. Cooperative game-theoretic approaches have been used for coalition formation \cite{bistaffa2017cooperative}, and network-based models aim to minimize fleet size while preserving service quality \cite{vazifeh2018addressing}. Preference-aware strategies further enhance match quality in dynamic systems \cite{bian2019mechanism}.

Recent research also leverages machine learning to improve operational efficiency. \citet{chen2016dynamic} show surge pricing boosts driver supply, and \citet{bongiovanni2022machine} develop a learning-based pipeline for autonomous ridesharing. Other studies address delay mitigation \cite{fielbaum2020unreliability}, real-time dispatching and repositioning \cite{riley2020real}, and rider satisfaction through behavior-aware modeling~\cite{yatnalkar2020enhanced}. Spatio-temporal forecasting and personalized behavior encodings further improve demand prediction \cite{tang2021multi, liu2021behavior2vector}.

\section{Conclusion}
\label{sec:conclusion}

In this paper, we introduced the problem of online deadhead control and formulated it as an optimization problem aimed at reducing the expected carbon emissions of ridesharing platforms while maintaining low rider wait times. We proposed \algName, an online algorithm designed to achieve near-optimal solutions by dynamically adjusting deadhead distance limits based on the number of ride requests in the assignment queue. Along with providing a theoretical analysis of \algName's performance relative to the optimal offline algorithm, we conducted extensive experiments using both synthetic and real-world datasets to evaluate its effectiveness. Our results show that \algName outperforms state-of-the-art algorithms across a variety of scenarios, with its advantages becoming particularly evident during high-demand periods. In future work, we plan to develop an online algorithm with a worst-case performance guarantee that minimizes both emissions and wait times while ensuring fair ride assignments across different drivers.


\section{Acknowledgment}
This research was supported in part by NSF grants CAREER 2045641, CPS-2136199, CNS-2106299, CNS-2102963, CSR-1763617, CNS-2106463, and CNS-1901137. We acknowledge their financial assistance in making this project possible.


\bibliographystyle{plainnat}
\bibliography{paper}


\end{document}